\documentclass[a4paper,UKenglish]{lipics}
\usepackage[ruled,vlined,commentsnumbered,titlenotnumbered]{algorithm2e}
\usepackage{algorithmic}
\usepackage{amsfonts}
\usepackage{amssymb}
\usepackage{amsmath}
\usepackage{amsthm}
\usepackage{latexsym}
\usepackage{enumerate}
\usepackage{tikz}
\usepackage{xifthen}

\theoremstyle{plain}

\newtheorem{conjecture}[theorem]{\textsf{\textbf{Conjecture}}}

\def \rTISP {\textsf{TISP}}
\def \wTISP {\textsf{word-TISP}}

\def \Z {\mathbb Z}
\def \R {\mathbb R}

\newenvironment{reminder}[1]{\smallskip
	\noindent {\textsf{\textbf{Reminder of #1}}}\em}{}

\title{Deterministic Time-Space Tradeoffs for k-SUM}
\titlerunning{Deterministic Time-Space Tradeoffs for k-SUM}

\author[1]{Andrea Lincoln\thanks{Supported by a Stanford Graduate Fellowship.}}
\author[2]{Virginia Vassilevska Williams\thanks{Supported by NSF Grants CCF-1417238, CCF-1528078 and CCF-1514339, and BSF Grant BSF:2012338.}}
\author[3]{Joshua R. Wang\thanks{Supported by a Stanford Graduate Fellowship.}}
\author[4]{R. Ryan Williams\thanks{Supported in part by NSF CCF-1552651. Any opinions, findings and conclusions or recommendations expressed in this material are those of the authors and do not necessarily reflect the views of the National Science Foundation.}}

\affil[1]{Computer Science Department, Stanford University, USA\\ \texttt{andreali@cs.stanford.edu}}
\affil[2]{Computer Science Department, Stanford University, USA\\ \texttt{virgi@cs.stanford.edu}}
\affil[3]{Computer Science Department, Stanford University, USA\\ \texttt{joshua.wang@cs.stanford.edu}}
\affil[4]{Computer Science Department, Stanford University, USA\\ \texttt{rrw@cs.stanford.edu}}

\authorrunning{A. Lincoln, V. Vassilevska Williams, J. R. Wang and R. R. Williams} 

\Copyright{Andrea Lincoln, Virginia Vassilevska Williams, Joshua R. Wang and R. Ryan Williams}

\EventEditors{Ioannis Chatzigiannakis, Michael Mitzenmacher, Yuval
	Rabani, and Davide Sangiorgi}
\EventNoEds{4}
\EventLongTitle{43rd International Colloquium on Automata, Languages,
	and Programming (ICALP 2016)}
\EventShortTitle{ICALP 2016}
\EventAcronym{ICALP}
\EventYear{2016}
\EventDate{July 11--15, 2016}
\EventLocation{Rome, Italy}
\EventLogo{eatcs}
\SeriesVolume{55}
\ArticleNo{XXX}
\begin{document}
\maketitle

\begin{abstract} Given a set of numbers, the $k$-SUM problem asks for a subset of $k$ numbers that sums to zero. When the numbers are integers,  the time and space complexity of $k$-SUM is generally studied in the word-RAM model; when the numbers are reals, the complexity is studied in the real-RAM model, and space is measured by the number of reals held in memory at any point.

We present a time and space efficient deterministic self-reduction for the $k$-SUM problem which holds for both models, and has many interesting consequences. To illustrate: 
 
\begin{itemize}
	
	\item $3$-SUM is in deterministic time $O(n^2 \lg\lg(n)/\lg(n))$ and space $O\left(\sqrt{\frac{n \lg(n)}{\lg\lg(n)}}\right)$. In general, 	any polylogarithmic-time improvement over quadratic time for $3$-SUM can be converted into an algorithm with an identical time improvement but low space complexity as well.
	
	\item $3$-SUM is in deterministic time $O(n^2)$ and space $O(\sqrt n)$, derandomizing an algorithm of Wang.

	\item A popular conjecture states that 3-SUM requires $n^{2-o(1)}$ time on the word-RAM. We show that the 3-SUM Conjecture is in fact equivalent to the (seemingly weaker) conjecture that every $O(n^{.51})$-space algorithm for $3$-SUM requires at least $n^{2-o(1)}$ time on the word-RAM.

	\item  For $k \ge 4$, $k$-SUM is in deterministic 
	$O(n^{k - 2 + 2/k})$ time and $O(\sqrt{n})$ space.
	
\end{itemize}

\end{abstract}

\newpage

\section{Introduction}

We consider the $k$-SUM problem: given a list $S$ of $n$ values, determine whether there are distinct $a_1,\ldots,a_k\in S$ such that $\sum_{i=1}^k a_i = 0$.
This classic problem is a parameterized version of the Subset Sum problem, which is among Karp's original NP-Complete problems.\footnote{Karp's definition of the Knapsack problem is essentially Subset Sum~\cite{karp1972reducibility}.} 

The brute-force algorithm for $k$-SUM runs in $O(n^k)$ time, and it is known~\cite{Patrascu:Ryan:10} that an $n^{o(k)}$ time algorithm (where the little-o depends on $k$) would violate the Exponential Time Hypothesis~\cite{ipz1}. A faster meet-in-the-middle algorithm reduces the $k$-SUM problem on $n$ numbers to $2$-SUM on $O(n^{\lceil k/2\rceil})$ numbers, which can then be solved by sorting and binary search in $O(n^{\lceil k/2\rceil}\log n)$ time. The belief that this meet-in-the-middle approach is essentially time-optimal is at the heart of many conditional $3$-SUM-hardness results in computational geometry (e.g.~\cite{gajentaan1995class}) and string matching (e.g.~\cite{jumbled,avwlocal}). 

The space usage of the meet-in-the-middle approach is prohibitive: the $O(n\log n)$ time solution for $2$-SUM uses linear space, which causes the fast $k$-SUM algorithm to need $\Omega(n^{\lceil k/2\rceil})$ space. However, the brute-force algorithm needs only $O(k\log n)$ space. This leads to the natural question: \emph{how well can one trade off time and space in solving $k$-SUM?} 

Schroeppel and Shamir~\cite{schroeppel1981t} first studied time-space tradeoff algorithms for Subset Sum. They showed how to reduce Subset Sum to an instance of $k$-SUM for any $k \geq 2$: split the elements into $k$ sets of $n/k$ elements each; for each set, compute $2^{n/k}$ sums corresponding to the subsets of the set; this forms a $k$-SUM instance of size $2^{n/k}$.
Since the $k$-SUM instance does not have to be explicitly stored, any time $T(N)$, space $S(N)$ algorithm for $k-$SUM immediately implies a time $T(2^{n/k})$, space $S(2^{n/k})$ algorithm for Subset Sum.
Furthermore, Schroeppel and Shamir gave a deterministic $\tilde{O}(n^2)$ time, $\tilde{O}(n)$ space algorithm for $4$-SUM, implying a $O^*(2^{n/2})$ time, $O^*(2^{n/4})$ space algorithm for Subset Sum.\footnote{The notation $\tilde{O}$ suppresses polylogarithmic factors in $n$, and $O^*$ suppresses polynomial factors in $n$.} They also generalized the algorithm to provide a smooth time-space tradeoff curve, with extremal points at $O^*(2^{n/2})$ time, $O^*(2^{n/4})$ space and $O^*(2^n)$ time, $O^*(1)$ space.

A recent line of work leading up to Austrin et al.~\cite{Austrin:2013} has improved this long-standing tradeoff curve for Subset Sum via \emph{randomized} algorithms, resulting in a more complex curve.  Wang~\cite{wang2014space} moved these gains to the $k$-SUM setting. In particular, for $3$-SUM he obtains an $\tilde{O}(n^2)$ time, $\tilde{O}(\sqrt n)$ space Las Vegas algorithm.

Despite the recent progress on the problem, all of the improved algorithms for the general case of $k$-SUM have heavily relied on randomization, either utilizing hashes or random prime moduli. These improvements also all rely heavily on the values in the lists being integers. For the general case of $k$-SUM, the previous best {\em deterministic} $k$-SUM results (even for integer inputs) are the brute-force algorithm, the meet-in-the-middle algorithm, and the Schroeppel and Shamir $4$-SUM algorithm, and simple combinations thereof. 

\subsection{Our Results}

We consider new ways of trading time and space in solving $k$-SUM, on both integer and real inputs (on the word-RAM and real-RAM respectively), without the use of randomization. Our improvements for $k$-SUM naturally extend to improvements to Subset Sum as well. 

Our main result is a deterministic self-reduction for $k$-SUM. Informally, we show how to deterministically decompose a list of $n$ numbers into a small collection of shorter lists, such that the $k$-SUM solution is preserved. This result is shown for $k=3$ in Section \ref{sec:3sum}. It is shown for general $k$ in Section \ref{sec:ksum}. 

\begin{theorem}\label{thm:main}
	Let $g$ be any integer between $1$ and $n$.  $k$-SUM on $n$ numbers can be reduced to $O(kg^{k-1})$ instances of $k$-SUM on $n/g$ numbers. The reduction uses $O(n g^{k-1})$ additional time and $O(n/g)$ additional words of space.
\end{theorem}

Theorem~\ref{thm:main} has several interesting applications. First, it leads to more efficient $k$-SUM algorithms. For example, Gold and Sharir, building on other recent advances, report a deterministic algorithm for $3$-SUM that works in both the word-RAM and real-RAM models and which runs in time $O(n^2 \lg\lg(n)/\lg(n))$ \cite{gold2015improved}. However, this algorithm uses a considerable amount of space to store a table of permutations. Applying Theorem~\ref{thm:main} in multiple ways and calling their algorithm, we recover the \emph{same} asymptotic running time but with drastically better space usage: 

\begin{theorem}
	There is an $O(n^2 \lg\lg(n)/\lg(n))$  time deterministic algorithm for $3$-SUM that stores at $O(\sqrt{\frac{n \lg(n)}{\lg\lg(n)}})$ numbers in memory at point. (An analogous statement holds for $3$-SUM over the integers.)
\end{theorem} 

Theorem~\ref{thm:main} also directly leads to a derandomization of Wang's space-efficient algorithm for $3$-SUM:

\begin{theorem}
	For all $s\in [0,1/2]$ there is a deterministic time $O(n^{3-2s})$, algorithm which uses $O(n^s)$ words of space for $3$-SUM.
\end{theorem}

From Theorem~\ref{thm:main} we can also derive a more space-efficient algorithm for $4$-SUM, and lift it to a new algorithm for $k$-SUM:

\begin{theorem}
	For $k \ge 4$, $k$-SUM is solvable in deterministic $O(n^{k - 2 + 2/(k-3)} )$ time and $O( \sqrt{n} )$ space in terms of words.
\end{theorem}

\paragraph*{A more plausible $3$-SUM conjecture.}
A rather popular algorithmic conjecture is the \emph{$3$-SUM Conjecture} that $3$-SUM on $n$ integers requires $n^{2-o(1)}$ time on a word-RAM with $O(\log n)$ bit words. This conjecture has been used to derive conditional lower bounds for a variety of problems~\cite{gajentaan1995class,jumbled,avwlocal,patrascu2010towards,popdyn}, and appears to be central to our understanding of lower bounds in low-polynomial time. To refute the conjecture, one could conceivably construct an algorithm that runs in $O(n^{1.99})$ time, but utilizes $\Omega(n^{1.99})$ space in some clever way. Here we consider a seemingly weaker (and thus more plausible) conjecture:

\begin{conjecture}[\textbf{The Small-Space 3-SUM Conjecture}] 

	On a word-RAM with $O(\log n)$-bit words, there exists an $\epsilon > 0$ such that every algorithm that solves $3$-SUM in $O(n^{1/2+\epsilon})$ space  must take at least $n^{2-o(1)}$ time.
	\label{conj:better3SUM}
\end{conjecture}

This conjecture looks weaker than the original $3$-SUM Conjecture, because we only have to prove a quadratic-time lower bound for all algorithms that use slightly more than $\sqrt{n}$ space. Proving time lower bounds is generally much easier when space is severely restricted (e.g. \cite{BeameSSV03,FortnowLMV05,DiehlMW11,Williams08,BeameCM13}).

Our self-reduction for $3$-SUM yields the intriguing consequence that the original $3$-SUM Conjecture is \emph{equivalent} to the Small-Space $3$-SUM conjecture! That is, the non-existence of a truly subquadratic-time $3$-SUM algorithm is equivalent to the non-existence of a truly subquadratic-time $n^{0.51}$-space $3$-SUM algorithm, even though the latter appears to be a more plausible lower bound. We prove:
\begin{theorem}
	If $3$-SUM is solvable in time $O(n^{2-\epsilon})$ time, then for every $\alpha > 0$ there is a $\delta > 0$ such that $3$-SUM is solvable in $O(n^{2-\delta})$ time and space $O(n^{1/2 + \alpha})$ in terms of words.
	\label{thm:intro3SUMconjsmallspace}
\end{theorem}

Theorem \ref{thm:intro3SUMconjsmallspace} is interesting, regardless of the veracity of the $3$-SUM conjecture. On the one hand, the theorem reduces the difficulty of proving the $3$-SUM Conjecture if it is true, because we only have to rule out small-space sub-quadratic time algorithms. On the other hand, the theorem means that refuting the $3$-SUM conjecture immediately implies a truly-subquadratic time algorithm for $3$-SUM using small space as well, which would be an algorithmic improvement.

\section{Preliminaries}

\subsection{$k$-SUM and Selection}

We will use the following version of the $k$-SUM problem:

\begin{definition}
	In the \emph{$k$-SUM problem}, we are given an unsorted list $L$ of $n$ values (over $\Z$ or $\R$) and want to determine if there are $a_1, \ldots, a_k \in L$ such that $\sum_{i=1}^k a_i = 0$.
\end{definition}

One fundamental case is the $3$-SUM problem. Sometimes $3$-SUM is presented with three separate lists, which we denote as $3$-SUM', but the two are reducible to each other in linear time, and with no impact on space usage.

\begin{definition}
	In the \emph{$3$-SUM problem}, we are given an unsorted list $L$ of $n$ values and want to know if there are $a,b,c\in L$ such that $a+b+c=0$. 	In the \emph{$3$-SUM' problem}, we are given three unsorted lists $A$, $B$, and $C$ of values, where $|A|=|B|=|C|=n$, and want to know if there are $a \in A, b \in B, c \in C$ such that $a+b+c=0$.  
\end{definition}

As part of our $k$-SUM algorithms, the classical \emph{Selection Problem} will also arise:

\newcommand{\kselect}[1][s]{$#1$-\textsc{Select}}

\begin{definition}
	In the \emph{\kselect{} problem}, we are given an unsorted list $L$ of $n$ values and a natural number $s$, and want to determine the $s^{th}$ smallest value in $L$.
\end{definition}

\subsection{Computational Model}

As standard when discussing sub-linear space algorithms, the input is provided in read-only memory, and the algorithm works with auxiliary read/write memory which counts towards its space usage. 

\emph{Computation on Integers.} When the input values are integers, we work in the word-RAM model of computation: the machine has a word size $w$, and we assume all input numbers can be represented with $w$ bits so that they fit in a word. Arithmetic operations ($+,-,*$) and comparisons on two words are assumed to take $O(1)$ time. Space is counted in terms of the number of words used.

\emph{Computation on Reals.} When the input values are real numbers, we work in a natural \emph{real-RAM} model of computation, which is often called the \emph{comparison-addition model} (see, for example, \cite{PettieR05}). Here, the machine has access to registers that can store arbitrary real numbers; addition of two numbers and comparisons on real numbers take $O(1)$ time. Space is measured in terms of the number of reals stored.

\emph{Time-Space Complexity Notation.} We say that $k$-SUM is solvable in
$\rTISP(T(n),S(n))$ if $k$-SUM on lists of length $n$ can be solved by a single algorithm running in deterministic $O(T(n))$ time and $O(S(n))$ space simultaneously on the real-RAM (and if the lists contain integers, on the word-RAM).

\subsection{Other Prior Work}
Baran, Demaine and Patrascu~\cite{baran2005subquadratic} obtained randomized slightly subquadratic time algorithms for Integer $3$-SUM in the word-RAM.  Gr{\o}nlund and Pettie~\cite{gronlund2014threesomes} studied $3$-SUM over the reals, presenting an $O(n^2/(\log n/\log\log n))$ time randomized algorithm, as well as a deterministic algorithm running in $O(n^2/(\log n/\log\log n)^{2/3})$ time. Recently, Gold and Sharir~\cite{gold2015improved} improved this deterministic running time to $O(n^2/(\log n/\log\log n))$. Abboud, Lewi and Williams~\cite{losingweight} showed that Integer $k$-SUM is W[1]-complete under randomized FPT reductions (and under some plausible derandomization hypotheses, the reductions can be made deterministic). In the linear decision tree model of computation, $k$-SUM over the reals is known to require $\Omega(n^{\lceil k/2\rceil})$ depth $k$-linear decision trees~\cite{Erickson95,AilonC05},
but the problem can be solved with $O(n^{k/2}\sqrt{\log n})$ depth $(2k-2)$-linear decision trees~\cite{gronlund2014threesomes}. The randomized decision tree complexity was improved by Gold and Sharir~\cite{gold2015improved} to $O(n^{k/2})$.

\section{Building Blocks}

In this section, we describe two tools we use to obtain our main self-reduction lemma for $k$-SUM and $3$-SUM. The first tool helps us guarantee that we don't have to generate too many subproblems in our reduction; the second will allow us to find these subproblems in a time and space efficient way.

\subsection{Domination Lemma}

Our deterministic self-reduction for $k$-SUM will split lists of size $n$ into $g$ sublists of size $n/g$, then solve subproblems made up of $k$-tuples of these sublists. Naively, this would generate $g^k$ subproblems to enumerate all $k$-tuples. In this section, we show that we only need to consider $O(k g^{k-1})$ subproblems. 

First, we define a partial ordering on $k$-tuples on $[n]^k$. For $t,t' \in [n]^k$, we say that $t \prec t'$ if $t[i] < t'[i]$ for all $i=1,\ldots,k$.	 (Geometrically, the terminology is that $t'$ \emph{dominates} $t$.) 

\begin{lemma}[Domination Lemma]
	Suppose all tuples in a subset $S \subseteq [n]^k$ are incomparable with respect to $\prec$. 	Then $|S| \le k n^{k-1}$.	\label{lem:domination}
\end{lemma}

The Domination Lemma can be seen as an extension of a result in \cite{vassilevska2009finding} (also used in~\cite{czumajlingas} in a different context) which covers the $k = 3$ case.

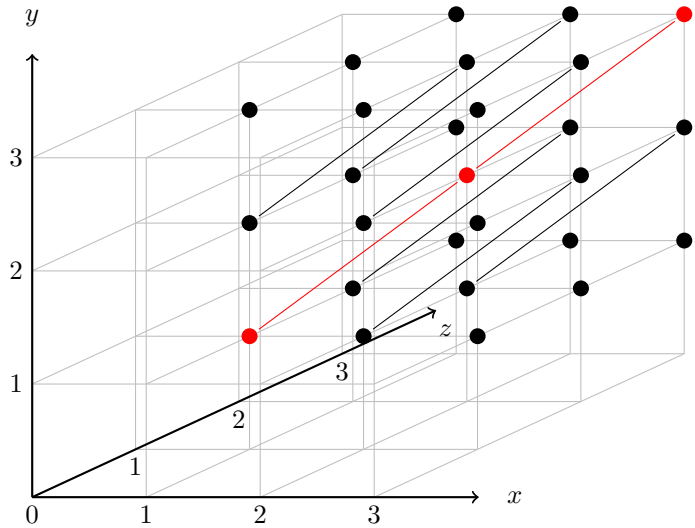
\begin{figure}
	\centering
	\begin{tikzpicture}[scale=1.5]
	\pgfmathsetmacro{\COS}{cos(25)}
	\pgfmathsetmacro{\SIN}{sin(25)}
	
	\foreach \x in {0, ..., 3} {
		\foreach \y in {0, ..., 3} {
			\foreach \z in {0, ..., 3} {
				\pgfmathsetmacro{\X}{\x+\z*\COS}
				\pgfmathsetmacro{\Y}{\y+\z*\SIN}
				\node (c\x\y\z) at (\X, \Y) {};
				\ifthenelse{\x<3}{
					\draw[black!25] (\X, \Y) -- (\X+1, \Y);
				}{}
				\ifthenelse{\y<3}{
					\draw[black!25] (\X, \Y) -- (\X, \Y+1);
				}{}
				\ifthenelse{\z<3}{
					\draw[black!25] (\X, \Y) -- (\COS+\X, \SIN+\Y);
				}{}
			}
		}
	}
	
	\node[label=right:$x$] (x-axis) at (4, 0) {};
	\node[label=above:$y$] (y-axis) at (0, 4) {};
	\node[label=below:$z$] (z-axis) at (4*\COS, 4*\SIN) {};
	\foreach \x in {0, ..., 3} {
		\node[below] at (\x, 0) {$\x$};
	}
	\foreach \y in {1, ..., 3} {
		\node[left] at (0, \y) {$\y$};
	}
	\foreach \z in {1, ..., 3} {
		\node[below] at (\z*\COS, \z*\SIN) {$\z$};
	}
	\draw[thick, ->] (0, 0) -- (x-axis);
	\draw[thick, ->] (0, 0) -- (y-axis);
	\draw[thick, ->] (0, 0) -- (z-axis);
	
	\fill[red] (c111) circle[radius=2pt];
	\fill[red] (c222) circle[radius=2pt];
	\fill[red] (c333) circle[radius=2pt];
	\draw[red] (c111) -- (c222) -- (c333);
	
	\foreach \x in {1, ..., 3} {
		\foreach \y in {1, ..., 3} {
			\foreach \z in {1, ..., 3} {
				\ifthenelse{\not\(\x = \y\) \OR \not\(\y = \z\)} {
					\fill (c\x\y\z) circle[radius=2pt];
					\ifthenelse{\x < 3 \AND \y < 3 \AND \z < 3} {
						\pgfmathtruncatemacro{\X}{\x+1}
						\pgfmathtruncatemacro{\Y}{\y+1}
						\pgfmathtruncatemacro{\Z}{\z+1}
						\draw (c\x\y\z) -- (c\X\Y\Z);
					}{}
				}{}
			}
		}
	}
	
	\end{tikzpicture}
	\caption{Domination Lemma chains when $n = 3$ and $k = 3$. The chain
		\{(1, 1, 1), (2, 2, 2), (3, 3, 3)\} is highlighted in red. In three
		dimensions, the number of chains is roughly proportional to the surface area
		of the cube, which is only $O(n^2)$, despite the fact that there are $O(n^3)$
		points.}
\end{figure}

\begin{proof}	We will give a cover of all elements in $[n]^k$ with few chains under $\prec$. Then by Dilworth's theorem, any set of incomparable elements under $\prec$ can only have one element from each chain.
	
	Take any $k$-tuple $t \in [n]^k$ such that $t[i] = 1$ for some $i=1,\ldots,k$. Letting $\ell \in [n]$ be the largest element in $t$, we define the chain $C(t)=\{t_0, t_1,\ldots,t_{n-\ell}\}$, where each $t_j$ is given by $t_j[i]=t[i]+j$ for all $i=1,\ldots,k$. Clearly $C(t)$ forms a chain in $[n]^k$ under $\prec$. Moreover these chains cover all elements of $[n]^k$: observe that the tuple $t$ appears in the chain $C(t')$ where $t'[i] = t[i] - \min_j t[j] + 1$ for all $i=1,\ldots,k$. 
	
	The number of chains is exactly the number of $k$-tuples with a $1$ in at least one coordinate. This number is less than $k$ times the number of tuples that have a $1$ in dimension $i$. The number of tuples with a $1$ in dimension $i$ is $n^{k-1}$. Thus, the total number of chains is $\leq kn^{k-1}$.	
\end{proof}

\begin{figure}[h!]
	\caption{A depiction of how $L$ is divided.}
	\centering
	
	\includegraphics[scale=0.35]{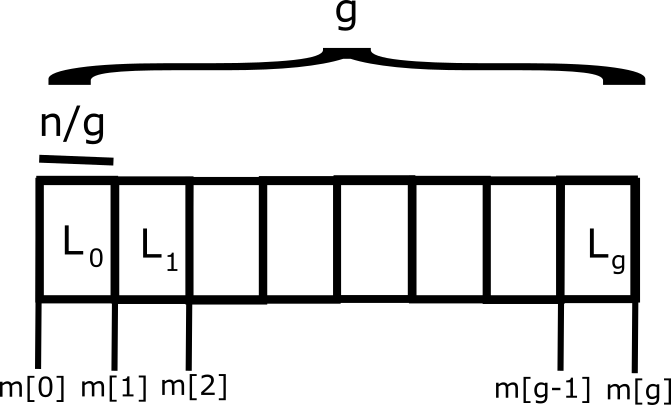}
	
\end{figure}

The Domination Lemma can be applied to show that in any list of numbers, not too many $k$-SUM subproblems can have $k$-SUM solutions. In the following, let $g$ divide $n$ for simplicity. Given a list $L$ of $n$ numbers divided into $g$ groups of size $n/g$, a \emph{subproblem of $L$} is simply the union of a $k$-tuple of groups from $L$. Note that a subproblem contains at most $kn/g$ numbers. 

\begin{corollary} Given a $k$-SUM instance $L$, suppose $L$ is divided into $g$ groups $L_1, \ldots, L_g$ where $|L_i| = n/g$ for all $i$, and for all $a \in L_i$ and $b \in L_{i+1}$ we have $a \leq b$. Then there are $O(k \cdot g^{k-1})$ subproblems $L'$ of $L$ such that the smallest $k$-sum of $L'$ is less than zero and the largest $k$-sum of $L'$ is greater than zero.	Furthermore, if some subproblem of $L$ has its largest or smallest $k$-sum equal to $0$, then the corresponding $k$-SUM solution can be found in $O(g^k)$ time.
	\label{cor:ksum-decomp}
\end{corollary}

\begin{proof}
	We associate each subproblem of $L$ with a corresponding $k$-tuple $(x_1,\ldots,x_k)\in [g]^k$ corresponding to the $k$ sublists $(L_{x_1}, \ldots, L_{x_k})$ of $L$.  
	
	Let $m[i]$ be the element in position $i\cdot (n/g)$ when $L$ is in sorted order. Consider any subproblem with $\sum_{i=1}^{k} m[x_i] > 0$ (smallest $k$-sum greater than zero) or $\sum_{i=1}^{k} m[x_{i}+1] < 0$ (largest $k$-sum less than zero). We call such a subproblem \emph{trivial}, since it cannot contain $k$-SUM solutions.

	In $O(g^k)$ time, we can determine whether any subproblem has $\sum_{i=1}^k m[x_i]=0$, and return the corresponding $k$-SUM solution if this is the case. Otherwise, we can assume that for each subproblem either it is trivial, or $\sum_{i=1}^k m[x_i] < 0 < \sum_{i=1}^k m[x_i+1]$. 
	
	Consider the set of non-trivial subproblems. Because for all $a \in L_i$ and $b \in L_{i+1}$ we have $a \leq b$, if for two subproblem $k$-tuples we have $t \prec t'$, then the smallest $k$-sum of the subproblem $t'$ is at least the largest $k$-sum of the subproblem $t$. This implies that at least one of the two subproblems must be trivial. In other words, the set of nontrivial problems corresponds to a set of incomparable $k$-tuples in $[g]^k$. Applying Lemma~\ref{lem:domination}, the number of nontrivial subproblems is $O(k g^k)$.
\end{proof}

\subsection{Bucket Retrieval and Space-Efficient Selection}
\label{subsec:nextgroup}

\newcommand{\nextgroup}[0]{\textsc{NextGroup}}

A randomized algorithm for $k$-SUM can partition a list of numbers by choosing a hash function at random, then loop over the hash function range to partition a given list into smaller buckets. Given a hash and a bucket number, it is easy to retrieve the contents of that bucket by scanning the list. 

To derandomize this process, we could try to create small ``hash'' buckets by grouping the $n/g$ smallest elements together, then the next $n/g$ smallest elements, and so on, \emph{without actually sorting the list}. However, retrieving the contents of a bucket may now be difficult to do with small space: we need to know the smallest and largest elements of a bucket to retrieve its elements, and we may not be able to store all of these extrema. We require an efficient algorithm to compute the largest element of a bucket, given the smallest element and the bucket size.

This problem is equivalent to the \emph{selection problem}, also known as \kselect{}, which asks for the $s^{th}$ smallest element of a list, when we set $s = n/g$. To reduce from our problem to \kselect{}, pretend that every entry less than our smallest element
is $\infty$. (To reduce from \kselect{} to our problem, we can pretend our smallest element is $-\infty$.) 

The classic median-of-median algorithm can solve \kselect{} in $O(n)$ time and $O(n)$ space~\cite{blum1973time}. Since we care about space usage, we provide an algorithm below which has $O(n)$ running time, but uses much less space. This algorithm turns out to be optimal for our purposes, since retrieving the bucket afterwards will already take $O(n)$ time and $O(s)$ space. 
\begin{lemma}
	\label{lem:nextgroup}
	\kselect{} can be solved in $O(n)$ time and $O(s)$ space.
\end{lemma}

\begin{proof}
	The plan is to scan through the elements of the list, inserting them to a data	structure $D$ which will allow us to track the smallest $s$ elements. We perform $n$ insertions, then query $D$ to ask for the smallest $s$ elements it contains. To get the claimed algorithm for selection, we give a data	structure can handle these operations in $O(1)$ amortized update time and $O(s)$	query time, with a data structure using only $O(s)$ space.
	
	One first attempt might be to build a heap of $s+1$ elements, which throws away the largest element whenever it gets full. Since heaps have logarithmic update time and linear space usage, this results in $O(\log s)$	update time, $O(s)$ query time, and $O(s)$ space.
	
	We can improve the update time by batching when we throw out large 	elements. Suppose instead we keep an array which can hold up to $2s$ elements. 	When the array gets full, we throw out the largest $s$ elements. To do this, we first compute the $(s+1)^{th}$ smallest element in the array. This can be done in $O(s)$ time and $O(s)$ space via the classical median-of-medians algorithm. We then do a linear scan of the array, and write all elements strictly less than the median to	 a new array. To handle ties, we write a copy of the median to the new array, 	until it has $s$ elements. When we are given our final query, we again throw	 out large elements so that we only have $s$ elements left, and then return those.
	
	Updates now take amortized constant time: after $s$ updates, we take $O(s)$ time to clear out the large elements. The	final query takes $O(s)$ time, since we again need to throw out large 	elements. The space usage is $O(s)$ since we store up to $2s$ elements, and running median-of-medians takes $O(s)$ space. This completes the proof.
\end{proof}

We will call the above algorithm $\textsc{NextGroup}$. $\textsc{NextGroup}$ takes as input a value $v$, a natural number $s$, and a list of numbers $L$, and outputs the next $s$ elements of $L$ in sorted order after the value $v$. Other variations on deterministic \kselect{} algorithms are mentioned in Appendix \ref{app:kselect}.

\section{Subquadratic $3$-SUM implies Subquadratic small-space $3$-SUM}
\label{sec:3sum}
We will begin by using our building blocks to prove a self reduction for $3$-SUM. Then we will show three intriguing consequences of this self reduction.  First, the self reduction can be used to show a general theorem that takes subquadratic algorithms for $3$-SUM and produces subquadratic time algorithms that run in nearly $\sqrt{n}$ space. Second, we show that algorithms for $3$-SUM that are subquadratic by polylog factors can be used to obtain $3$-SUM algorithms with the same asymptotic running time and simultaneously small  space. Finally, we will prove that the Small-Space $3$-SUM conjecture is equivalent to the $3$-SUM conjecture.

\subsection{$3$-SUM Self Reduction}

We now proceed to solve $3$-SUM using our bucket retrieval subroutine. We will
use $\max S$ and $\min S$ to refer to the maximum and minimum elements of a list $S$, respectively.

As anticipated, we split the three arrays into groups of size $n/g$, and solve
$3$-SUM on subproblems of this size. Naively there are $O(g^3)$ subproblems to
solve, but we use Corollary~\ref{cor:ksum-decomp} to argue we only get $O(g^2)$
subproblems.

\begin{theorem}[3-SUM Self-Reduction Theorem]
	\label{thm:3sum-real}
	If $3$-SUM is solvable in $\rTISP(T(n), S(n))$ then for any $g$, $3$-SUM can be solved in
	$\rTISP(g^2 (n+T(n/g)), n/g + S(n/g))$.
\end{theorem}

\begin{proof}
	Consider the following algorithm.
	
	\begin{algorithm}[H]
		\label{alg:3sum-real}
		Set $preva = -\infty$\;
		\For{$i \in [0, g-1]$}{
			Set $A' = \textsc{NextGroup}(A, preva, n/g + 1)$\;
			Set $prevb = -\infty$\;
			\For{$j \in [0, g-1]$}{
				Set $B' = \textsc{NextGroup}(B, prevb, n/g + 1)$\;
				Set $C' = \textsc{NextGroup}(C, - \max A' - \max B', n/g + 1)$\;
				\While{$\min C' \le - \min A' - \min B'$}{
					\If{$3$-SUM($A', B', C'$) returns true}{
						\Return true\;
					}
					Set $C' = \textsc{NextGroup}(C, \max C', n/g + 1)$\;
				}
				Set $prevb = \max B'$\;
			}
			Set $preva = \max A'$\;
		}
		\Return false\;
		\caption{$3$-SUM Algorithm}
	\end{algorithm}
	
	Algorithm~\ref{alg:3sum-real} is correct because we consider all possible
	elements of $C$ where the sum of elements from $A'$ and $B'$ could land, and
	the choices of $A'$ and the choices of $B'$ cover all of $A$ and $B$,
	respectively. If there are multiple copies of a value in a list we will fail to
	list all copies only if it already appeared in a previous sublist. This will not affect
	correctness because the value will have already been analyzed.
	
	It's easy to see that the algorithm calls $\textsc{NextGroup}$ $O(g)$ times
	for $A'$, $O(g^2)$ times for $B'$. We claim that we also only call it $O(g^2)$
	times for $C'$. To show this, we want to apply
	Corollary~\ref{cor:ksum-decomp}. Unfortunately, the groups of $C$ that we
	extract don't always line up with our ideal $n/g$ division; since we start at
	$- \max A' - \max B'$, we may not align at the endpoints of blocks.
	Fortunately, we've only introduced an extra $O(1)$ possibilities of $C'$ for
	every $(A', B')$ pair, or $O(g^2)$ extras total. Hence we still only make
	$O(g^2)$ calls to $\textsc{NextGroup}$. By Lemma~\ref{lem:nextgroup}, these
	calls will require $O(ng^2)$ time and $O(n/g)$ space.
	
	Our algorithm also calls the $\rTISP(T(n), S(n))$ algorithm for $3$-SUM
	$O(g^2)$ times on instances of size $O(n/g)$, which requires $O(g^2T(n/g))$
	time and $O(S(n/g))$ space.
	
	We have shown Algorithm~\ref{alg:3sum-real} is correct and has the desired
	runtime and space usage, so this completes the proof.
\end{proof}

\subsection{General Theorem for Space Reduction}

Our self-reduction for $3$-SUM yields the following intriguing consequence: \emph{subquadratic-time algorithms for $3$-SUM imply 
	subquadratic-time small-space algorithms for $3$-SUM}. Plugging this connection into known $3$-SUM algorithms, we can automatically obtain more space-efficient $3$-SUM algorithms for free. From a complexity-theoretic point of view, the consequence is perhaps even more intriguing: it means that the $3$-SUM Conjecture is \emph{equivalent} to the statement that there is no subquadratic-time $n^{0.51}$-space $3$-SUM algorithm, even though the latter appears to be a more plausible lower bound(!).

We begin by stating our generic space reduction theorem.

\begin{theorem}[$3$-SUM Space Reduction] \label{thm:spacereduction} Suppose $3$-SUM is solvable in $n^2/f(n)$ time, where $1 \leq f(n) \leq n$. Then $3$-SUM is solvable by an algorithm running in $O(n^2/f(n/g))$ time and $O(n/h)$ space simultaneously, where $g(n),h(n) \in [1,n]$ satisfy the relations \[g(n) \geq \Omega \left(\sqrt{\frac{n}{f(n/g(n))}}\right) \textrm{~~~and~~~} h^2 + \frac{n}{f(n/(h g(n/h)))} \leq O\left(\frac{n}{f(n/g(n))}\right).\]
\end{theorem} 
\begin{proof} We will apply our Self-Reduction Theorem for $3$-SUM (Theorem~\ref{thm:3sum-real}) in two different ways. First, we will use the self-reduction (and the constraint on $g(n)$) to convert our $3$-SUM algorithm into a linear-space algorithm, with a modest increase in running time (if at all). Pushing the linear-space algorithm through the self-reduction once more will reduce the space bound further, without increasing the running time asymptotically (using the constraint on $h(n)$).
	
	Let $T(n) := n^2/f(n)$. Set the parameter $g(n) \geq 1$ to satisfy  
	\begin{equation}\label{eqn:setg} T(n/g) = \frac{n^2/g^2}{f(n/g)} = O(n); \textrm{~~or, equivalently~~} g = \Omega \left( \sqrt{\frac{n}{f(n/g)}} \right).\end{equation} 
	Assuming $g$ satisfies \eqref{eqn:setg}, applying the $3$-SUM Self-Reduction (Theorem~\ref{thm:3sum-real}) with $T(n) = S(n)$ and $g$, we can then solve $3$-SUM in \begin{equation}\label{eqn:linspace}\rTISP\left(g^2(n+T(n/g)), n/g + T(n/g)\right) = \rTISP \left(\frac{n^2}{f(n/g)}, n\right).\end{equation}
	
	Now, set new time and space bounds $T(n) := n^2/f(n/g(n))$, $S(n) = n$ from \eqref{eqn:linspace}. Then, applying the $3$-SUM Self-Reduction (Theorem~\ref{thm:3sum-real}) with the new $T(n)$, $S(n)$ and some parameter $h$, we can then solve $3$-SUM in $\rTISP(h^2(n+T(n/h)), n/h + S(n/h)) =$
	\begin{equation*}\label{eqn:linspaces} \rTISP\left(h^2 \left(n+\frac{n^2/h^2}{f(n/(h g(n/h)))}\right),n/h\right) \subseteq \rTISP \left(\frac{n^2}{f(n/g)}, n/h\right),\end{equation*} by our hypothesis on $h$.
\end{proof}

\subsection{Space-Efficient Fast $3$-SUM}

When we apply Theorem~\ref{thm:spacereduction} directly to known algorithms, we obtain immediate space improvements with negligible loss in running time. Very recently, Gold and Sharir~\cite{gold2015improved} have given a faster $3$-SUM algorithm in the real-RAM model, building on the work of Gronlund and Pettie~\cite{gronlund2014threesomes}:

\begin{theorem}[Gold and Sharir~\cite{gold2015improved}]
	3-SUM can be solved in $O(n^2 \lg\lg(n)/ \lg(n))$ time over the reals and integers. \label{thm:gold3SUM}
\end{theorem}

As discussed in the introduction, their novel approach uses quite a bit of space. Applying Theorem~\ref{thm:spacereduction}, we can reduce the space usage to only $O \left(\sqrt{n\lg(n)/\lg\lg(n)}\right)$, with the same asymptotic running time of Gold and Sharir. 

\begin{corollary}[Space-Efficient $3$-SUM Algorithm] \label{cor:savelogs3sum}
	3-SUM is in $\rTISP \left(n^2 \frac{\lg\lg(n)}{\lg(n)}, \sqrt{\frac{n \lg(n)}{\lg\lg(n)}}\right)$.
\end{corollary}

\begin{proof} We shall apply Theorem~\ref{thm:spacereduction}. First, set $f(n) := \lg(n)/\lg\lg(n)$, so that $3$-SUM is solvable in $O(n^2/f(n))$ time by Theorem~\ref{thm:gold3SUM}. 
	
	Set $g(n) := \sqrt{\frac{n \lg\lg(n)}{\lg(n)}}$ and $h(n) := \sqrt{\frac{n \lg\lg(n)}{\lg(n)}}$. By our choice of $f(n)$ and basic properties of logarithms, observe that \begin{equation}\label{eqn:ng}f(n/g) = f(\tilde{O}(\sqrt{n})) = \Theta(f(n)), \end{equation} and furthermore  
	\begin{equation}\label{eqn:nhg}f(n/(h g(n/h))) = f\left(\tilde{O}(\sqrt{n})/\Tilde{O}(n^{1/4})\right) = \Theta(f(n)).\end{equation} 
	
	By \eqref{eqn:ng}, we have \[g = \sqrt{\frac{n \lg\lg(n)}{\lg(n)}} \geq \Omega \left(\sqrt{\frac{n}{f(n/g)}}\right), \textrm{~so the first constraint of Theorem~\ref{thm:spacereduction} is satisfied.}\] Moreover, by \eqref{eqn:nhg} we have 
	\[
	h^2 + \frac{n}{f(n/(h g(n/h)))} = \frac{n \lg\lg(n)}{\lg(n)} + \frac{n}{\Theta(f(n))}, \textrm{~which is~} O \left(\frac{n}{f(n/g)}\right) \textrm{~by \eqref{eqn:ng}}.\]
	Therefore the second constraint of Theorem~\ref{thm:spacereduction} is also satisfied, and $3$-SUM is solvable by an algorithm running in  $O(n^2/f(n))$ time and $O\left(\sqrt{n f(n)} \right)$ space simultaneously.  
\end{proof} 

In general, Theorem~\ref{thm:spacereduction} provides a generic \emph{reduction} from faster $3$-SUM algorithms to faster space-efficient $3$-SUM algorithms. To illustrate:

\begin{corollary} If 3-SUM is solvable in $O(n^2/\lg^a(n))$ time for some constant $a > 0$, then $3$-SUM is in $\rTISP(n^2/\lg^a(n), \sqrt{n}\lg^{a/2}(n))$.
\end{corollary}
\begin{proof} We apply Theorem~\ref{thm:spacereduction}. By assumption we have $3$-SUM in $O(n^2/f(n))$ time, where $f(n) = \lg^a n$. Set $g(n) := \sqrt{n}/ \lg^{a/2}(n)$, and $h(n) :=  \sqrt{n}/\lg^{a/2}(n)$. Note that $f(n/g(n)) = \Theta(\lg^a(n))$ and $f(n/(h(n)\cdot g(n/h(n)))) = \Theta(\lg^a(n))$, similar to Corollary~\ref{cor:savelogs3sum}. Therefore
	
	\[g(n) = \sqrt{n}/ \lg^{a/2}(n) \geq \Omega \left(\sqrt{\frac{n}{f(n/g(n))}}\right)\] and \[h^2 + \frac{n}{f(n/(h g(n/h)))} \leq O\left( \frac{n}{\lg^a n}\right) \leq O\left(\frac{n}{f(n/g(n))}\right).\] Hence Theorem~\ref{thm:spacereduction} applies to these settings of the parameters, and $3$-SUM is in $O(n^2/f(n/g)) = O(n^2/\lg^a(n))$ time and $O(n/h) = O(\sqrt{n}\lg^{a/2}(n))$ space.	  	  	  
\end{proof}

\subsection{The $3$-SUM Conjecture and Small Space}

Finally, we use the Space Reduction Theorem (Theorem~\ref{thm:spacereduction}) to show that the $3$-SUM conjecture is false, then it is also false with respect to small-space algorithms.


\begin{lemma}
	If $3$-SUM is in $O(n^{2-\epsilon})$ time for some $\epsilon > 0$, then for every $\alpha >0 $, there is a $\delta > 0$ such that $3$-SUM is solvable in $O(n^{2-\delta})$ time and $O(n^{1/2+\alpha})$ space, simultaneously.
	\label{lem:3sumfasterepsilons}
\end{lemma}

\begin{proof} 
	
	The proof of Theorem~\ref{thm:spacereduction} applies the $3$-SUM Self Reduction (Theorem~\ref{thm:3sum-real}) twice. We will basically perform the first part of the proof of Theorem~\ref{thm:spacereduction}, but instead of applying the second part of the proof, we have to choose a different setting of parameters, focused on minimizing the space usage instead of preserving running time. 
	
	Let $T(n) := n^2/f(n)$ with $f(n) = n^{\epsilon}$. We first reduce the space usage of the algorithm to linear. To this end, set $g(n) := n^{(1-\epsilon)/(2-\epsilon)}$. Then, applying the $3$-SUM Self-Reduction (Theorem~\ref{thm:3sum-real}) with $T(n) = S(n)$ and $g(n)$, we can then solve $3$-SUM in \begin{equation*}\rTISP(g^2(n+T(n/g)), n/g + T(n/g)) = \rTISP \left(\frac{n^2}{f(n/g)}, n\right)
	= \rTISP \left(\frac{n^2}{n^{\epsilon/(2-\epsilon)}}, n\right).\end{equation*}
	
	Now reset $f(n) := n^{\epsilon/(2-\epsilon)}$, and reset $g(n) := n^{1/2+\alpha}$ with $\alpha \in (0,1/2)$. Applying the $3$-SUM Self-Reduction (Theorem~\ref{thm:3sum-real}) with $T(n) = n^2/f(n)$, $S(n) = n$, and $g(n)$ as above, 
	we find an algorithm for $3$-SUM in 
	$$\rTISP\left(n^{2-2\alpha}+ n^{2- (1/2-\alpha)\epsilon/(2-\epsilon)},  n^{1/2+\alpha}\right).$$
	Note that for all $\epsilon > 0$ and $\alpha \in (0,1/2)$, the running time bound is truly subquadratic. Further note that for any $\alpha \geq 1/2$, we only have more space to work with, so we clearly obtain $O(n^{2-\delta})$ time and $O(n^{1/2+\alpha})$ space (for some $\delta > 0$) in that case as well.
\end{proof}

This lemma can be applied to show that the $3$-SUM Conjecture is equivalent to seemingly much weaker statement:

\begin{reminder}{The Small-Space $3$-SUM Conjecture (Conjecture~\ref{conj:better3SUM})} 
	On a word-RAM with $O(\log n)$-bit words, there exists an $\epsilon > 0$ such that every algorithm that solves $3$-SUM in $O(n^{1/2+\epsilon})$ space  must take at least $n^{2-o(1)}$ time.
\end{reminder}

\begin{theorem}
	The Small-Space $3$-SUM Conjecture is equivalent to the $3$-SUM Conjecture. 
\end{theorem}
\begin{proof}
	It suffices to show that the $3$-SUM Conjecture if true implies the Small-Space $3$-SUM Conjecture and that the refutation of the $3$-SUM Conjecture implies the Small-Space $3$-SUM Conjecture. First, we observe that the $3$-SUM Conjecture trivially implies the Small-Space $3$-SUM Conjecture. 
	
	Suppose the $3$-SUM Conjecture is false. Then a $O(n^{2-\epsilon})$ time algorithm for $3$-SUM exists, and Lemma \ref{lem:3sumfasterepsilons} implies that for every $\alpha > 0$, there is a $\delta > 0$ such that $3$-SUM is solvable in $O(n^{2-\delta})$ time and $O(n^{1/2+\alpha})$ space, simultaneously. But this means that for any choice of $\epsilon'>0$ for the Small-Space $3$-SUM Conjecture, we can find a truly-subquadratic $3$-SUM algorithm that uses only $O(n^{1/2+\epsilon'/2})$ space. This would falsify the Small-Space $3$-SUM Conjecture. 
\end{proof}

We conclude that, in order to prove the $3$-SUM conjecture, it is sufficient to prove that no algorithm can solve $3$-SUM in $\rTISP(n^{2-\epsilon}, n^{0.51})$ for some $\epsilon>0$.

\section{$k$-SUM}
\label{sec:ksum}
\subsection{$k$-SUM Self-Reduction}

We now generalize from $3$-SUM to $k$-SUM. Again, we plan to split the lists
into $g$ groups of size $O(n/g)$. By Corollary~\ref{cor:ksum-decomp}, we will
have only $O(g^{k-1})$ subproblems of size $O(n/g)$. Unlike $3$-SUM, where we
just used the naive algorithm to solve subproblems, in this section we use a
general algorithm; we reduce from $k$-SUM to itself (albeit on smaller instances).

\begin{theorem}
	\label{thm:real-self-reduction}
	Suppose real $k$-SUM can be solved in $\rTISP(T(n), S(n))$. Then for any $g$,
	it can also be solved in $\rTISP(g^{k-1}(n + T(n/g)), n/g + S(n/g))$.
\end{theorem}

\begin{proof}
	This follows from a generalized analysis of the proof of
	Theorem~\ref{thm:3sum-real}. We brute force over which groups the first $k-1$
	elements are in. We then extract groups where the negative sum of elements
	from these first $k-1$ groups could land. By Corollary~\ref{cor:ksum-decomp}
	and similar reasoning as before, there are only $O(g^{k-1})$ tuples of blocks.
	For each tuple, we make a call to $\textsc{NextGroup}$ and to our input
	$k$-SUM algorithm on a subproblem of size $O(n/g)$. This gives the desired
	time and space, completing the proof.
\end{proof}

\subsection{Applying our $k$-SUM Self-Reduction}
We want to apply the self-reduction on efficient deterministic algorithms. One of the best starting points is the Schroeppel-Shamir $4$-SUM algorithm, which we note is actually deterministic and works on reals because it simply uses priority queues and reduces to the classic $2$-SUM algorithm, both of which only use comparisons.
\begin{lemma}[From \cite{schroeppel1981t}]
	\label{lem:ksum-linearspace}
	Real $4$-SUM is solvable in 
	$\rTISP(n^2, n)$.
\end{lemma}

Another useful fact observed by Wang is that an algorithm for $k$-SUM can be transformed into an algorithm for $(k+1)$-SUM by brute-forcing one element:
\begin{lemma}[From \cite{wang2014space}]
	If Real $k$-SUM is solvable in 
	$\rTISP(T(n), S(n))$  then 
	real $k+1$-SUM is solvable in  $\rTISP(nT(n), S(n)+1)$.
	\label{lem:bootstrap}
\end{lemma}

Suppose we want to use our results to derive a linear-space algorithm for $k$-SUM. We will assume $k$ is a multiple of $4$, although Lemma~\ref{lem:bootstrap} allows us to fill in for the other values of $k$. By writing down sums of $k/4$ elements, we can transform $k$-SUM to $4$-SUM, yielding a $\rTISP(n^{k/2}, n^{k/4})$ algorithm. We can then apply Theorem~\ref{thm:real-self-reduction} with $g = n^{(k - 4)/k}$ to get a $\rTISP(n^{k - 3 + 4/k}, n )$ algorithm. Notice that this algorithm runs significantly faster than $O(n^k)$ time; we get $O(n^{11/2})$ for $8$-SUM and $O(n^{28/3})$ for $12$-SUM. As a coarse upper bound, we can apply Lemma~\ref{lem:bootstrap} and round down our savings (to make things cleaner), compensating for $k$ which are not a multiple of $4$, we get:

\begin{corollary}
	For $k \ge 4$, $k$-SUM  is solvable in
	$\rTISP(n^{k - 3 + 4/(k-3)}, n )$.
\end{corollary}

Suppose we wanted to use $O(\sqrt{n})$ space instead. We get smaller subproblems by making more groups; choosing $g = n^{(k-2)/k}$ instead yields a $\rTISP(n^{k - 2 + 2/k}, \sqrt{n} )$. Similarly applying Lemma~\ref{lem:bootstrap} and round down our savings to compensate for $k$ which are not a multiple of $4$, we get another coarse upper bound:

\begin{corollary}
	For $k \ge 4$, $k$-SUM  is solvable in
	$\rTISP(n^{k - 2 + 2/k}, \sqrt{n} )$.
\end{corollary}

\section{Future Work}

We would like to extend these results to derandomize other known randomized algorithms for $k$-SUM. To do that, it seems we require a ``deterministic simulation'' of the hash functions used in those results. Baran, Demaine, and Patrascu use hashing to get subquadratic algorithms for $3$-SUM \cite{baran2005subquadratic}; Patrascu uses it to reduce $3$-SUM to Convolution $3$-SUM \cite{patrascu2010towards}; Wang uses it to produce a family of linear-space algorithms for $k$-SUM \cite{wang2014space}. Which of these results, if any, can be derandomized? 

The hash families involved have three crucial properties: load-balancing (the hash buckets are not ``too large''), few subproblems (the number of $k$-tuples of hash buckets examined is ``small''), and few false positives (there are few non-$k$-SUM solutions mapped to $k$-tuples of hash buckets examined). Our \kselect{} algorithm (Lemma~\ref{lem:nextgroup}) and Domination Lemma (Lemma~\ref{lem:domination}) are used to achieve the first two properties, without using randomization. Can the last property also be simulated deterministically? (Note that it's not entirely clear what it would mean to simulate ``few false positives'' deterministically.) If so, it is likely that all these results can be derandomized efficiently.

\bibliographystyle{abbrv}
\bibliography{mybib} 

\begin{thebibliography}{10}

\bibitem{losingweight}
A.~Abboud, K.~Lewi, and R.~Williams.
\newblock Losing weight by gaining edges.
\newblock In {\em Algorithms - {ESA} 2014 - 22th Annual European Symposium,
  Wroclaw, Poland, September 8-10, 2014. Proceedings}, pages 1--12, 2014.

\bibitem{avwlocal}
A.~Abboud, V.~{Vassilevska~Williams}, and O.~Weimann.
\newblock Consequences of faster alignment of sequences.
\newblock In {\em Automata, Languages, and Programming - 41st International
  Colloquium, {ICALP} 2014, Copenhagen, Denmark, July 8-11, 2014, Proceedings,
  Part {I}}, pages 39--51, 2014.

\bibitem{popdyn}
A.~Abboud and V.~V. Williams.
\newblock Popular conjectures imply strong lower bounds for dynamic problems.
\newblock In {\em 55th {IEEE} Annual Symposium on Foundations of Computer
  Science, {FOCS} 2014, Philadelphia, PA, USA, October 18-21, 2014}, pages
  434--443, 2014.

\bibitem{AilonC05}
N.~Ailon and B.~Chazelle.
\newblock Lower bounds for linear degeneracy testing.
\newblock {\em J. {ACM}}, 52(2):157--171, 2005.

\bibitem{jumbled}
A.~Amir, T.~M. Chan, M.~Lewenstein, and N.~Lewenstein.
\newblock On hardness of jumbled indexing.
\newblock In {\em Automata, Languages, and Programming - 41st International
  Colloquium, {ICALP} 2014, Copenhagen, Denmark, July 8-11, 2014, Proceedings,
  Part {I}}, pages 114--125, 2014.

\bibitem{Austrin:2013}
P.~Austrin, P.~Kaski, M.~Koivisto, and J.~M{\"{a}}{\"{a}}tt{\"{a}}.
\newblock Space-time tradeoffs for subset sum: An improved worst case
  algorithm.
\newblock In {\em Automata, Languages, and Programming - 40th International
  Colloquium, {ICALP} 2013, Riga, Latvia, July 8-12, 2013, Proceedings, Part
  {I}}, pages 45--56, 2013.

\bibitem{baran2005subquadratic}
I.~Baran, E.~D. Demaine, and M.~Patra{\c{s}}cu.
\newblock Subquadratic algorithms for 3sum.
\newblock In {\em Algorithms and Data Structures}, pages 409--421. Springer,
  2005.

\bibitem{BeameCM13}
P.~Beame, R.~Clifford, and W.~Machmouchi.
\newblock Element distinctness, frequency moments, and sliding windows.
\newblock In {\em FOCS}, pages 290--299, 2013.

\bibitem{BeameSSV03}
P.~Beame, M.~E. Saks, X.~Sun, and E.~Vee.
\newblock Time-space trade-off lower bounds for randomized computation of
  decision problems.
\newblock {\em J. {ACM}}, 50(2):154--195, 2003.

\bibitem{blum1973time}
M.~Blum, R.~W. Floyd, V.~Pratt, R.~L. Rivest, and R.~E. Tarjan.
\newblock Time bounds for selection.
\newblock {\em Journal of computer and system sciences}, 7(4):448--461, 1973.

\bibitem{czumajlingas}
A.~Czumaj and A.~Lingas.
\newblock Finding a heaviest triangle is not harder than matrix multiplication.
\newblock In {\em Proc. SODA}, pages 986--994, 2007.

\bibitem{DiehlMW11}
S.~Diehl, D.~van Melkebeek, and R.~Williams.
\newblock An improved time-space lower bound for tautologies.
\newblock {\em J. Comb. Optim.}, 22(3):325--338, 2011.

\bibitem{Erickson95}
J.~Erickson.
\newblock Lower bounds for linear satisfiability problems.
\newblock In {\em Proceedings of the Sixth Annual {ACM-SIAM} Symposium on
  Discrete Algorithms, 22-24 January 1995. San Francisco, California.}, pages
  388--395, 1995.

\bibitem{FortnowLMV05}
L.~Fortnow, R.~J. Lipton, D.~van Melkebeek, and A.~Viglas.
\newblock Time-space lower bounds for satisfiability.
\newblock {\em J. {ACM}}, 52(6):835--865, 2005.

\bibitem{gajentaan1995class}
A.~Gajentaan and M.~H. Overmars.
\newblock On a class of {O}$(n^2)$ problems in computational geometry.
\newblock {\em Computational geometry}, 5(3):165--185, 1995.

\bibitem{gold2015improved}
O.~Gold and M.~Sharir.
\newblock Improved bounds for 3sum, k-sum, and linear degeneracy.
\newblock {\em arXiv preprint arXiv:1512.05279}, 2015.

\bibitem{gronlund2014threesomes}
A.~Gronlund and S.~Pettie.
\newblock Threesomes, degenerates, and love triangles.
\newblock In {\em Foundations of Computer Science (FOCS), 2014 IEEE 55th Annual
  Symposium on}, pages 621--630. IEEE, 2014.

\bibitem{ipz1}
R.~Impagliazzo and R.~Paturi.
\newblock On the complexity of k-{SAT}.
\newblock {\em J. Comput. Syst. Sci.}, 62(2):367--375, 2001.

\bibitem{karp1972reducibility}
R.~M. Karp.
\newblock {\em Reducibility among combinatorial problems}.
\newblock Springer, 1972.

\bibitem{patrascu2010towards}
M.~Patrascu.
\newblock Towards polynomial lower bounds for dynamic problems.
\newblock In {\em Proceedings of the forty-second ACM symposium on Theory of
  computing}, pages 603--610. ACM, 2010.

\bibitem{PettieR05}
S.~Pettie and V.~Ramachandran.
\newblock A shortest path algorithm for real-weighted undirected graphs.
\newblock {\em {SIAM} J. Comput.}, 34(6):1398--1431, 2005.

\bibitem{Patrascu:Ryan:10}
M.~P\u{a}tra\c{s}cu and R.~Williams.
\newblock On the possibility of faster sat algorithms.
\newblock In {\em Proceedings of the Twenty-First Annual ACM-SIAM Symposium on
  Discrete Algorithms}, SODA '10, pages 1065--1075, Philadelphia, PA, USA,
  2010. Society for Industrial and Applied Mathematics.

\bibitem{schroeppel1981t}
R.~Schroeppel and A.~Shamir.
\newblock A ${T}={O}(2^{n/2})$, ${S}={O}(2^{n/4})$ algorithm for certain
  np-complete problems.
\newblock {\em SIAM journal on Computing}, 10(3):456--464, 1981.

\bibitem{vassilevska2009finding}
V.~Vassilevska and R.~Williams.
\newblock Finding, minimizing, and counting weighted subgraphs.
\newblock In {\em Proceedings of the forty-first annual ACM symposium on Theory
  of computing}, pages 455--464. ACM, 2009.

\bibitem{wang2014space}
J.~R. Wang.
\newblock Space-efficient randomized algorithms for k-sum.
\newblock In {\em Algorithms-ESA 2014}, pages 810--829. Springer, 2014.

\bibitem{Williams08}
R.~R. Williams.
\newblock Time-space tradeoffs for counting {NP} solutions modulo integers.
\newblock {\em Computational Complexity}, 17(2):179--219, 2008.

\end{thebibliography}

\appendix

\section{$s$-Select}
\label{app:kselect}
In addition to  $\textsc{NextGroup}$ we have two other \kselect{} algorithms. We present two algorithms to solve this subtask. The first requires the values to be integers in the range $[-R, R]$ and runs in $\wTISP(n \log R, 1)$ (recall we are in the word-RAM model and we are measuring space in terms of the number of words). The other needs no assumptions and returns the answers for $g$ choices of $k$ in $\rTISP(n^2, g)$. The \nextgroup algorithm discussed in subsection \ref{subsec:nextgroup} runs in $\rTISP(n, s)$.

\subsection{Bounded Range $s$-Select}

This first algorithm runs a binary search over the bounded range to locate the $s^{th}$ smallest element.

\begin{algorithm}[H]
	\label{alg:bounded-range-kselect}
	Set $\ell = -R$, $r = R$\;
	\While{$\ell < r$} {
		Set $m = \lfloor \frac{\ell + r}{2} \rfloor$\;
		Set $c$ = 0\;
		\For{$a \in L$} {
			\If{$m \ge a$} {
				Increment $c$\;
			}
		}
		\eIf{$c \ge s$} {
			Set $r = m$\;
		} {
		Set $\ell = m+1$\;
	}
}
\Return $\ell$\;
\caption{Bounded Range \kselect{} Algorithm}
\end{algorithm}

\begin{theorem}
	\label{thm:kselect-range}
	Algorithm~\ref{alg:bounded-range-kselect} solves \kselect{} in $\wTISP(n \log R, 1)$.
\end{theorem}

\begin{proof}
	Algorithm~\ref{alg:bounded-range-kselect} returns the smallest integer $v$ such that there are $s$ values less than or equal to $v$. Since all values are integers, by assumption, this is the $s^{th}$ smallest value. The algorithm runs for $O(\log R)$ iterations, but each iteration does a scan of $L$ that takes $O(n)$ time. The algorithm keeps a constant number of values, so it uses $O(1)$ space.
\end{proof}

\subsection{Batch real $k$-Select}

When we lose the range and integrality assumptions, we can still gain when we have several \kselect{} instances with the same list $L$. In particular, suppose there are $g$ indices we want to know: $s_1, \ldots, s_g$, where $g \le n$, we can go through the list in order in $n^2$ time noting and saving the value of all of those indices. Furthermore, we can use this method over the reals.

\begin{algorithm}[H]
	\label{alg:batch-kselect}
	Set $prev = -\infty$\;
	Create a return vector $V$ of length $g$\;
	Set $i = 1$\;
	\While{$i \le n$} {
		Set $curr = \infty$\;
		\For{$a \in L$} {
			If $a > prev$, set $curr = \min (curr, a)$\;
		}
		Set $dup = 0$\;
		\For{$a \in L$} {
			If $a = curr$, increment $dup$\;
		}
		\For{$j = [1, g]$} {
			If $k_j \in [i, i + dup)$, set $V[j] = curr$\;
		}
		Set $prev = curr$\;
		Increment $i$ by $dup$\;
	}
	\Return $V$\;
	\caption{Batch \kselect{} Algorithm}
\end{algorithm}

\begin{theorem}
	Algorithm~\ref{alg:batch-kselect} solves batch real \kselect{} in $\rTISP(n^2, g)$.
\end{theorem}

\begin{proof}
	Algorithm~\ref{alg:batch-kselect} repeatedly scans $L$, each time finding the next largest element. After it finds the $s^{th}$ smallest element, it checks to see if $s$ was one of the requested indices, and if so, fills it into its answer. The algorithm performs $O(n)$ scans of $L$ and the $k_j$, but since $g \le n$, this runs in $O(n^2)$ time. Keeping $g$ elements around takes $O(g)$ space.
\end{proof}

\end{document}